\documentclass{article}
\usepackage{Sub/spconf,amsmath,graphicx,hyperref}

\usepackage{balance}
\usepackage{epstopdf}

\usepackage{amsmath,amssymb,epsfig,cite}
\include{macros}
\usepackage{graphicx}
\usepackage{epstopdf}

\usepackage{acronym}

\usepackage{bm}

\usepackage{pifont}

\usepackage{array}

\usepackage{tikz}
\usetikzlibrary{calc}

\usepackage{tcolorbox}

\usepackage{accents}

\allowdisplaybreaks

\usepackage[caption=false,font=footnotesize]{subfig}

\usepackage{soul}

\usepackage{amsthm}

\usepackage{nicefrac}

\usepackage{lipsum,multicol}

\usepackage{algorithm}
\usepackage{algpseudocode}

\algdef{SE}[SUBALG]{Indent}{EndIndent}{}{\algorithmicend\ }%
\algtext*{Indent}
\algtext*{EndIndent}

\usepackage{enumitem}
\usepackage{mwe}    
\usepackage{epsfig}
\usepackage{color}
\usepackage{url}
\usepackage{mathtools,xparse}

\usepackage{gensymb}

\usepackage[multiple]{footmisc}

\usepackage{xr}

\makeatletter
\newcommand*{\addFileDependency}[1]{
  \typeout{(#1)}
  \@addtofilelist{#1}
  \IfFileExists{#1}{}{\typeout{No file #1.}}
}
\makeatother

\newcommand*{\myexternaldocument}[1]{%
    \externaldocument{#1}%
    \addFileDependency{#1.tex}%
    \addFileDependency{#1.aux}%
}

\myexternaldocument{supp}

\usepackage{xcolor}

\newcommand\absss[1]{\left|#1\right|}

\newcommand\abss[1]{\lvert#1\rvert}
\newcommand\norm[1]{\left\lVert #1\right\rVert}

\newcommand\normss[1]{\lVert #1\rVert}

\newcommand{\thn}[1]{ {#1^{\rm{th} } } }

\newcommand{\Eee}{\mathbb{E}}

\newcommand{\sigmaphi}{\sigma^2_{\phi}}

\newcommand{\rxii}{R_{\xi \xi}}
\newcommand{\deltat}{\Delta t}

\newcommand{\sigmaxi}{\sigma^2_{\xi}}

\newcommand{\Tsym}{ T_{\rm{sym}} }
\newcommand{\Ts}{ T_{\rm{s}} }
\newcommand{\Tcp}{ T_{\rm{cp}} }

\newcommand{\FF}{ \mathbf{F} }

\newcommand{\deltaf}{ \Delta f }

\newcommand{\fc}{ f_c }

\newcommand{\Nrx}{ R }

\newcommand{\ptx}{ \pp_{\rm{T}} }
\newcommand{\prxi}{ \pp^{i}_{\rm{R}} }
\newcommand{\distt}{ d_{\rm{T}} }
\newcommand{\distri}{ d^{i}_{\rm{R}} }
\newcommand{\gaint}{ G_{\rm{T}} }
\newcommand{\gainr}{ G_{\rm{R}} }

\newcommand{\alphab}{ \boldsymbol{\alpha}}
\newcommand{\gammab}{ \boldsymbol{\gamma}}

\newcommand{\taumin}{ \tau_{\rm{min}} }

\newcommand{\stilde}{ \widetilde{s} }

\newcommand{\ytilde}{ \widetilde{y} }

\newcommand{\mtN}{{\mathcal{N}}}

\newcommand{\fdb}{ f_{\rm{3dB}} }

\newcommand{\rect}[1]{ { \rm{rect} }\left(#1\right) }

\newcommand{\mub}{ \boldsymbol{\mu} }
\newcommand{\mubbar}{ \widebar{\boldsymbol{\mu}} }

\newcommand{\mtCN}{{\mathcal{CN}}}

\newcommand{\tracesmall}[1]{ {{{\rm{tr}}\left( #1 \right)}}  }

\newcommand{\realp}[1]{ \Re \left\{#1\right\}  }
\newcommand{\imp}[1]{ \Im \left\{#1\right\}  }

\newcommand{\Imatrix}{{ \boldsymbol{\mathrm{I}} }}

\newcommand{\bb}{ \mathbf{b} }

\newcommand{\boldzero}{{ {\boldsymbol{0}} }}
\newcommand{\boldone}{{ {\boldsymbol{1}} }}

\newcommand{\Pt}{{ P_{\rm{T}} }}

\newcommand{\sigmarcsi}{{ \sigma_{{\rm{rcs}},i} }}

\newcommand{\boldR}{ \mathbf{R} }

\newcommand{\yy}{ \mathbf{y} }

\newcommand{\pp}{ \mathbf{p} }
\newcommand{\xx}{ \mathbf{x} }

\newcommand{\ww}{ \mathbf{w} }
\newcommand{\zz}{ \mathbf{z} }

\newcommand{\uu}{ \mathbf{u} }

\newcommand{\bxi}{ \boldsymbol{\xi} }

\newcommand{\trp}{\mathsf{T}}
\newcommand{\herm}{\mathsf{H}}

\newcommand{\cset}[2]{ \mathbb{C}^{#1 \times #2}  }
\newcommand{\csett}{ \mathbb{C}  }
\newcommand{\rset}[2]{ \mathbb{R}^{#1 \times #2}  }
\newcommand{\rsett}{ \mathbb{R}  }

\newcommand{\alphabar}{ \widebar{\alpha} }
\newcommand{\ppbar}{ \widebar{\pp} }
\newcommand{\varthetabar}{ \widebar{\vartheta} }

\newcommand{\gammabar}{ \widebar{\gamma} }

\newcommand{\etab}{ {\boldsymbol{\eta}} }

\makeatletter \renewcommand\d[1]{\ensuremath{%
		\;\mathrm{d}#1\@ifnextchar\d{\!}{}}}
\makeatother

\makeatletter
\newcommand*\rel@kern[1]{\kern#1\dimexpr\macc@kerna}
\newcommand*\widebar[1]{%
  \begingroup
  \def\mathaccent##1##2{%
    \rel@kern{0.8}%
    \overline{\rel@kern{-0.8}\macc@nucleus\rel@kern{0.2}}%
    \rel@kern{-0.2}%
  }%
  \macc@depth\@ne
  \let\math@bgroup\@empty \let\math@egroup\macc@set@skewchar
  \mathsurround\z@ \frozen@everymath{\mathgroup\macc@group\relax}%
  \macc@set@skewchar\relax
  \let\mathaccentV\macc@nested@a
  \macc@nested@a\relax111{#1}%
  \endgroup
}
\makeatother

\theoremstyle{remark}

\newtheoremstyle{mytheoremstyle} 
    {\topsep}                    
    {\topsep}                    
    {\upshape}                   
    {.5em}                           
    {\itshape}                   
    {.}                          
    {.5em}                       
    {}  

\theoremstyle{plain}

\newtheoremstyle{iremark}
  {\topsep}   
  {\topsep}   
  {\upshape}  
  {0.2in}       
  {\itshape}  
  {.}         
  {5pt plus 1pt minus 1pt} 
  {\thmname{#1}\thmnumber{ \itshape#2}\thmnote{ (#3)}} 

\newtheorem{proposition}{Proposition}

\theoremstyle{definition}
\newtheorem*{remark}{Remark}
\acrodef{RIS}{reconfigurable intelligent surface}
\acrodef{SNR}{signal-to-noise ratio}
\acrodef{ISAC}{integrated sensing and communication}
\acrodef{ISLAC}{integrated sensing, localization, and communication}
\acrodef{LOS}{line-of-sight}
\acrodef{NLOS}{non-line-of-sight}
\acrodef{AOA}{angle-of-arrival}
\acrodef{AOD}{angle-of-departure}
\acrodef{UE}{user equipment}
\acrodef{NF}{near-field}
\acrodef{BS}{base station}
\acrodef{AP}{access point}
\acrodef{AF}{ambiguity function}
\acrodef{MCRB}{misspecified Cram\'{e}r-Rao bound}
\acrodef{CRB}{Cram\'{e}r-Rao bound}
\acrodef{LB}{lower bound}
\acrodef{ML}{maximum-likelihood}
\acrodef{MML}{mismatched maximum-likelihood}
\acrodef{DL}{downlink}
\acrodef{UL}{uplink}
\acrodef{VRU}{vulnerable road user}
\acrodef{MIMO}{multiple-input multiple-output}
\acrodef{MISO}{multiple-input single-output}
\acrodef{SISO}{single-input single-output}
\acrodef{SIP}{shift invariance property}
\acrodef{FIM}{Fisher information matrix}
\acrodef{RMSE}{root mean-squared error}
\acrodef{AWGN}{additive white Gaussian noise}
\acrodef{ADMM}{alternating direction method of multipliers}
\acrodef{LS}{least-squares}
\acrodef{SOC}{second-order cone}
\acrodef{CFO}{carrier frequency offset}
\acrodef{i.i.d.}{independently and identically distributed}
\acrodef{MI}{mutual information}
\acrodef{SAC}[S\&C]{sensing and communication}
\acrodef{OTFS}{orthogonal time frequency space}

\acrodef{ULA}{uniform linear array}
\acrodef{PSLR}{peak-to-sidelobe level ratio}

\acrodef{LO}{local oscillator}
\acrodef{PN}{phase noise}
\acrodef{CP}{cyclic prefix}
\acrodef{ISI}{inter-symbol interference}

\acrodef{FRO}{free-running oscillator}
\acrodef{PLL}{phase-locked loop}
\acrodef{PDF}{probability density function}
\acrodef{RCS}{radar cross section}
\acrodef{KL}{Kullback-Leibler}

\acrodef{PSD}{power spectral density}
\acrodef{TX}{transmit}
\acrodef{RX}{receive}

\acrodef{CDF}{cumulative distribution function}
\acrodef{SQP}{sequential quadratic programming}

\acrodef{MF}{matched filtering}
\acrodef{RF}{reciprocal filtering}

\acrodef{LMMSE}{linear minimum-mean-squared error}

\allowdisplaybreaks

\usepackage{etoolbox}
\makeatletter
\AtBeginDocument{%
  \patchcmd{\@makecaption}{\vskip 10pt}{\vskip 3pt}{}{}%
}
\renewcommand\section{\@startsection{section}{1}{\z@}%
  {-1.5ex plus -.3ex minus -.2ex}{.7ex plus .2ex}{\normalfont\Large\bfseries}}
\renewcommand\subsection{\@startsection{subsection}{2}{\z@}%
  {-1.2ex plus -.3ex minus -.2ex}{.5ex plus .2ex}{\normalfont\large\bfseries}}
\apptocmd{\thebibliography}{\footnotesize\setlength{\itemsep}{0pt}\setlength{\parsep}{0pt}\setlength{\baselineskip}{9.7pt}}{}{}
\makeatother
\title{Phase-Coherent Multistatic Sensing under Phase Noise \\ in Distributed ISAC}

\name{%
\shortstack[c]{%
Musa Furkan Keskin$^{\ast}$,
Silan Karadag$^{\ast,\dagger}$,
Navid Amani$^{\ast}$,\\
Nikita Lyamin$^{\star}$,
Annelie Wyholt$^{\dagger}$,
Henk Wymeersch$^{\ast}$}%
\thanks{This work, conducted within the Advanced Digitalization
program at the WiTECH Centre DisCouRSe, has been financed by VINNOVA, and, in part, by the Swedish Research Council (VR) through the project 6G-PERCEF under Grant 2024-04390.}%
}

\address{$^{\ast}$ \normalsize Department of Electrical Engineering, Chalmers University of Technology, Sweden \\ $^{\dagger}$ \normalsize Volvo Car Corporation, Gothenburg, Sweden,
 $^{\star}$ \normalsize Waveium AB, Malmö, Sweden}

\begin{document}
\ninept
\setlength{\abovedisplayskip}{3pt plus 1pt minus 1pt}
\setlength{\belowdisplayskip}{3pt plus 1pt minus 1pt}
\setlength{\abovedisplayshortskip}{1pt plus 1pt}
\setlength{\belowdisplayshortskip}{2pt plus 1pt minus 1pt}
\setlength{\jot}{2pt}
%
\maketitle
\begin{abstract}
We investigate phase-coherent multistatic sensing in distributed ISAC systems where spatially separated access points (APs) share a common, non-ideal local oscillator (LO) impaired by phase noise (PN). In the presence of a single transmit AP and multiple receive APs with OFDM transmission, we derive the joint statistics of PN across receive APs and reveal its unique covariance structure that depends on target location through propagation delays of individual bistatic links. Based on this statistical characterization, we carry out a \ac{CRB} analysis of three receiver models, namely PN-free, PN-aware and PN-ignorant, for both isotropic and non-isotropic targets. Simulation results verify the location-dependent statistics of PN in shared-LO distributed architectures while demonstrating how transmit power, LO quality, AP deployment and scattering characteristics affect multistatic localization.

\end{abstract}
%

\section{Introduction}
\label{sec_intro}
Distributed \ac{ISAC} extends environmental perception from individual transceivers to a network of cooperating \acp{AP} \cite{collaborative_PMN_WCM_2023,cooperative_ISAC_WCM_2025}. Spatially separated \acp{AP} observe targets from complementary viewpoints, providing geometric diversity and broader sensing coverage \cite{cooperative_ISAC_WCM_2025}. By fusing sensing information at different levels, ranging from raw measurements to local estimates, distributed \ac{ISAC} can improve localization accuracy in multistatic configurations relative to conventional monostatic sensing \cite{DISAC_2025_Kawon}. By leveraging the existing communication infrastructure, it can thus overcome key limitations of single-node sensing and support the demanding sensing requirements envisioned for 6G applications.

In phase-synchronized distributed \ac{ISAC}, further gains can be achieved through phase-coherent processing, which preserves the carrier phase relationships across \acp{AP} and allows their spatial separation to form an extremely large effective sensing aperture \cite{MIMO_radar_2008,multistatic_network_Dario_2024}. This becomes particularly attractive in bandwidth-limited sub-6 GHz deployments, where individual \acp{AP} may be equipped with only a small number of antennas. Carrier phase information has been exploited for uplink localization with distributed arrays \cite{vukmirovic2019direct,fascista2024joint} and for multistatic sensing \cite{phase_coherent_DMIMO_ISAC_2025,silan_SPAWC_2026}, demonstrating its potential for high accuracy and resolution. A key practical requirement to realize these gains is to maintain inter-AP phase coherence. 

Such coherence requires either a shared \ac{LO} \cite{CF_MIMO_PN_2021,wu2024uplink} or accurate synchronization and calibration of independently generated signals \cite{massiveSync_Erik_TSP_2024}. In both architectures, \ac{PN}, resulting from \ac{LO} imperfections, can significantly degrade communication performance due to its rapidly time-varying nature \cite{wu2024uplink,ganesan2023beamsync,ngo2026distributed}. However, its implications for distributed sensing remain relatively underexplored. In standalone radar systems, the impact of \ac{PN} has been studied for shared- and independent-\ac{LO} multistatic architectures using FMCW and PMCW waveforms \cite{Duerr_multistatic_PN_JMW_2022,distributed_PMCW_PN_2024}. For OFDM-based ISAC, prior work has characterized \ac{PN}-induced performance loss and delay-dependent \ac{PN} statistics in \textit{monostatic} sensing \cite{PN_OFDM_ISAC_RadarConf_2023,OFDM_PN_TSP_Exploitation_2023,de2024sensing_PN_OFDM}, while \textit{bistatic} studies have considered hybrid monostatic-bistatic processing \cite{Kawon_Sync_PN_DISAC_RadarConf25}, clutter-assisted \ac{PN} mitigation \cite{idrees2026phase} and sensing privacy \cite{PN_Sensing_Privacy_ISAC_2026}. 
However, the existing literature lacks a unified characterization of \ac{PN} in shared-\ac{LO} \textit{multistatic} architectures for OFDM-based distributed \ac{ISAC} systems. In such a setting, \ac{PN} exhibits a distinctive statistical property arising from the fact that each sensing receiver observes the same oscillator process with a link-specific delay. As a result, the differential/self-referenced \ac{PN} \cite{Demir_PN_2006,PN_2006} across receivers becomes correlated through their respective bistatic propagation delays (similar to the range correlation effect in monostatic sensing \cite{Range_Correlation_93}), which implies that the joint space-time \ac{PN} statistics depend on the unknown target location via multiple bistatic delays. Despite recent efforts on the analysis of \ac{PN} in monostatic and bistatic sensing \cite{PN_OFDM_ISAC_RadarConf_2023,OFDM_PN_TSP_Exploitation_2023,de2024sensing_PN_OFDM,Kawon_Sync_PN_DISAC_RadarConf25,idrees2026phase}, a rigorous theoretical characterization of this unique, \textit{location-dependent} statistical structure as well as its impact on phase-coherent multistatic location estimation have remained unexplored.

To address this gap, we investigate OFDM-based phase-coherent multistatic sensing in a distributed \ac{ISAC} system with one transmit \ac{AP} and multiple receive \acp{AP} sharing a common, non-ideal \ac{LO} impaired by \ac{PN}. Our contributions are twofold. First, under a Wiener \ac{PN} model, we derive the joint statistics of the differential \ac{PN} samples across time and receive \acp{AP}, yielding a target-location-dependent covariance matrix. Second, building on this characterization, we derive deterministic, hybrid and misspecified \acp{CRB} for PN-free, PN-aware and PN-ignorant localization, respectively, considering both isotropic and non-isotropic targets. The proposed framework enables a systematic investigation of how \ac{PN} impacts multistatic localization across different SNRs, \ac{LO} qualities, target locations and scattering characteristics in shared-\ac{LO} architectures.

\section{System Model}
\subsection{Scenario Description}
We consider a distributed ISAC system consisting of a \ac{TX} AP and $\Nrx$ sensing \ac{RX} APs, all equipped with a single antenna. The APs share a common \ac{LO} \cite{CF_MIMO_PN_2021,wu2024uplink} and therefore are time, frequency and phase synchronized. Due to hardware imperfections, the \ac{LO} is assumed to be non-ideal and impaired by \ac{PN} \cite{OFDM_PN_TSP_Exploitation_2023,wu2024uplink}, leading to the \ac{PN} process $\phi(t)$. We assume the presence of a single point target in the environment. In a 2D global coordinate system, the TX AP and the $\thn{i}$ RX AP have known and fixed locations denoted by $\ptx$ and $\prxi$ while the unknown target location is denoted by $\pp$. The TX AP sends an OFDM signal for downlink communications while the RX APs collect and jointly process the target echoes for phase-coherent multistatic sensing \cite{phase_coherent_DMIMO_ISAC_2025}. 

\subsection{Transmit Signal Model}
The TX AP transmits an OFDM symbol over $N$ subcarriers, with subcarrier spacing $\deltaf=1/T$, useful duration $T$, \ac{CP} duration $\Tcp$ and total duration $\Tsym=\Tcp+T$ \cite{RadCom_Proc_IEEE_2011}. The complex baseband transmit signal is given by $s(t)=(1/\sqrt{N})\sum_{n=0}^{N-1}x_n e^{j2\pi n\deltaf t}\rect{t/\Tsym}$, where $x_n$ denotes the complex data symbol on subcarrier $n$, known to the RX APs. In the presence of PN, the upconverted transmit signal is \cite{PN_2006,OFDM_PN_TSP_Exploitation_2023}
\begin{equation}\label{eq_passband_st}
\stilde(t) = \Re \left\{  s(t) e^{j \left[ 2 \pi \fc t + \phi(t)\right]} \right\} \,,
\end{equation} 
where $\fc$ denotes the carrier frequency.

\subsection{Receive Signal Model at Sensing RX APs}\label{sec_radar_rec}
The signal originating from the TX and arriving at RX $i$ after being scattered off the target experiences a bistatic delay $\tau_i(\pp) = (\norm{\ptx - \pp} + \norm{\prxi - \pp})/c$, where $c$ denotes the speed of propagation. Then, given the transmit signal in \eqref{eq_passband_st}, the passband target echo at RX $i$ is \cite{OFDM_PN_TSP_Exploitation_2023}
\begin{align}\nonumber
    \ytilde_i(t) = \Re \left\{ \alpha_i \,  e^{j \vartheta_i} s(t-\tau_i(\pp)) e^{j \left[ 2 \pi \fc (t-\tau_i(\pp)) + \phi(t-\tau_i(\pp))\right]} \right\} \,,
\end{align}
where $\alpha_i \in \rsett$ represents the amplitude of the combined channel gain of the two-segment path, including the effect of transmit power, \ac{RCS} and the overall path loss \cite[Eq.~(7)]{NF_Coh_Sens_2022}, \cite[Eq.~(11)]{Li_24_ISAC_Toward_Seamless_Sensing}, and $\vartheta_i$ denotes the scattering-induced phase shift. Downconverting the passband signal through the shared, noisy \ac{LO}, i.e., multiplication by $e^{-j (2 \pi \fc t + \phi(t))}$ \cite{PN_SI_TSP_2017,OFDM_PN_TSP_Exploitation_2023}, yields the complex baseband signal 
\begin{align}\label{eq_rec_baseband}
    y_i(t)  
    &= \alpha_i \, e^{j \vartheta_i} s(t-\tau_i(\pp)) e^{-j 2 \pi \fc \tau_i(\pp)}   w_i(t, \pp) \,,
\end{align} 
where the multiplicative \ac{PN} process is denoted by $w_i(t, \pp) = e^{-j\xi(t, \tau_i(\pp))}$ with 
\begin{align}\label{eq_xi_p}
    \xi(t, \tau) \triangleq \phi(t) - \phi(t-\tau) \,.
\end{align}

For \ac{CP} removal in \eqref{eq_rec_baseband}, we consider the samples at $t_{\ell}=t_0+\ell T/N$ for $\ell=0,\ldots,N-1$, where $t_0=\taumin+\Tcp$ and $\taumin \triangleq \min_i \normss{\ptx-\prxi}/c$ denotes the earliest possible bistatic propagation time. To guarantee orthogonality of subcarriers, 
all bistatic delays of interest are assumed to lie within one CP of the earliest arrival, i.e., $0\leq \tau_i(\pp)-\taumin\leq \Tcp$ for all $i$ \cite{RadCom_Proc_IEEE_2011,navid_WCNC_2026}. Accordingly, the vector observation at RX $i$ is given by \cite{OFDM_PN_TSP_Exploitation_2023}
\begin{align}\label{eq_rx_vector_i}
    \yy_i
    =
    \alpha_i e^{j\vartheta_i}
    e^{-j2\pi f_c\tau_i(\pp)}
    \ww_i \odot \FF_N^\herm
    \big[\xx\odot\bb(\tau_i(\pp))\big]
    +\zz_i \,,    
\end{align}
where $\yy_i\in\cset{N}{1}$, $\xx \triangleq [x_0,\ldots,x_{N-1}]^{\trp}$, $\bb(\tau) \in\cset{N}{1}$ with $[\bb(\tau)]_n \triangleq e^{-j2\pi n\deltaf\tau}$, $\ww_i \in\cset{N}{1}$ contains the time-domain PN samples with $ [\ww_i]_\ell \triangleq  e^{-j\xi(t_\ell, \tau_i(\pp))}$, $\FF_N \in\cset{N}{N}$ is the unitary DFT matrix and $\zz_i \sim\mtCN(\boldzero,\sigma^2\Imatrix)$ is \ac{AWGN}. In \eqref{eq_rx_vector_i}, we distinguish between \textit{isotropic} and \textit{non-isotropic} scattering, leading to different observation models.
\begin{itemize}
    \item \textbf{Isotropic Scattering:} In this case, $\vartheta_i = \vartheta$, which is a common assumption in multistatic sensing \cite{NF_Coh_Sens_2022,Isotropic_2006,multistatic_NF_ISAC_2024,multistatic_network_Dario_2024,MIMO_radar_2008,darioCoherentMovingTarget_2025} and supported by experimental evidence \cite{SAR_isotropic_2023}. Then, \eqref{eq_rx_vector_i} becomes 
    \begin{align}\label{eq_rx_vector_i_iso}
    \yy_i    =    \alpha_i e^{j\vartheta}
    e^{-j2\pi f_c\tau_i(\pp)}
    \ww_i \odot \uu_i(\pp)
    +\zz_i \,,    
\end{align}
where $\uu_i(\pp) \triangleq \FF_N^\herm
    \big[\xx\odot\bb(\tau_i(\pp))\big] \in \cset{N}{1}$. In \eqref{eq_rx_vector_i_iso}, the carrier phase $2\pi f_c\tau_i(\pp)$ provides location information (on top of frequency-domain phase shifts in $\bb(\tau_i(\pp))$) via joint processing of $\{\yy_i\}_{i=1}^{\Nrx}$ since $\vartheta$ is common across the RX APs. 

    \item \textbf{Non-isotropic Scattering:} In this case, $\vartheta_i$ is an arbitrary unknown phase shift specific to each RX and therefore can be absorbed into an unknown complex channel gain $\gamma_i = \alpha_i e^{j\vartheta_i}
    e^{-j2\pi f_c\tau_i(\pp)} \in \csett$ since joint processing cannot separate the carrier phases from the scattering-induced phase shifts. Then,  \eqref{eq_rx_vector_i} can be reparameterized into 
    \begin{align}\label{eq_rx_vector_i_non_iso}
    \yy_i    =    \gamma_i
    \ww_i \odot \uu_i(\pp)
    +\zz_i \,,
\end{align}
where, unlike \eqref{eq_rx_vector_i_iso}, $\pp$ can only be estimated through $\bb(\tau_i(\pp))$.
\end{itemize}

\subsection{Problem Formulation}\label{sec_prob_form}
Given the observations $\{\yy_i\}_{i=1}^{\Nrx}$ in \eqref{eq_rx_vector_i_iso} or \eqref{eq_rx_vector_i_non_iso} collected from $\Nrx$ sensing RX APs, the goal is to estimate the target location $\pp$ under PN in the common \ac{LO}, which manifests itself through the PN vectors $\{\ww_i\}_{i=1}^{\Nrx}$. As seen from \eqref{eq_rec_baseband} and \eqref{eq_xi_p},
the joint statistics of $\{\ww_i\}_{i=1}^{\Nrx}$ depend on $\{\tau_i(\pp)\}_{i=1}^{\Nrx}$, which are coupled through the common unknown parameter $\pp$. Hence, the phase-coherent multistatic sensing problem involves dealing with \textit{location-dependent PN statistics}, which are \textit{a-priori} unknown, complicating the estimation process. In Sec.~\ref{sec_pn}, we derive the \ac{PN} statistics, followed by a \ac{CRB} analysis of three receiver models (PN-free, PN-aware and PN-ignorant) in Sec.~\ref{sec_bounds} to evaluate the impact of \ac{PN} on multistatic localization accuracy for both isotropic and non-isotropic targets.

\section{Phase Noise Statistics}\label{sec_pn}
In this section, we characterize the statistics of $\bxi = [\bxi_1^\trp \, \ldots \, \bxi_\Nrx^\trp]^\trp \in \rset{N\Nrx}{1}$, where $\bxi_i \triangleq [\xi(t_0, \tau_i(\pp)), \ldots, \xi(t_{N-1}, \tau_i(\pp))]^\trp$. According to \eqref{eq_rec_baseband}--\eqref{eq_rx_vector_i}, we have $\ww = e^{-j \bxi}$, where $\ww \triangleq [\ww_1^\trp \, \ldots \, \ww_\Nrx^\trp]^\trp$. 

\subsection{Statistics of Differential Phase Noise Process}
We adopt the \ac{FRO} model for the \ac{LO}, leading to the Wiener \ac{PN} process $\phi(t) \sim \mtN(0, \sigmaphi(t)) $, where $\sigmaphi(t) = 4 \pi \fdb t$ and $\fdb$ denotes the $3 \, \rm{dB}$ bandwidth of the Lorentzian oscillator spectrum \cite{Demir_PN_2000,PN_OFDM_PLL_TCOM_2007}.
Then, the differential/self-correlated PN process in \eqref{eq_xi_p} becomes a stationary process with \textit{delay-dependent} statistics \cite{Demir_PN_2006,PN_2006,OFDM_PN_TSP_Exploitation_2023}
\begin{equation}\label{eq_dpn_stat}
\xi(t, \tau) \sim \mtN(0, \sigmaxi(\tau)) \,,
\end{equation}
where $\sigmaxi(\tau) =  4 \pi \fdb \abss{\tau}$ is the delay-dependent variance.
We provide the following result regarding the statistics of $\xi(t, \tau)$.

\begin{proposition}\label{lemma_pn_stat_general}
    The correlation function of $\xi(t, \tau)$ in \eqref{eq_dpn_stat} across different time instances and delays, with $\deltat \triangleq t_1 - t_2$, is given by
    \begin{align} \label{eq_xixi}
        &\rxii(\deltat, \tau_i, \tau_j) \triangleq    \Eee\big[ \xi(t_1,\tau_i)\xi(t_2,\tau_j)\big]
        \\ \nonumber
    =&\tfrac{1}{2}\big[\sigmaxi(\Delta+\tau_j)+\sigmaxi(\Delta-\tau_i)-\sigmaxi(\Delta)-\sigmaxi(\Delta+\tau_j-\tau_i)\big] \,.
    \end{align}
\end{proposition}
\begin{proof}
    It follows by inserting \cite[Eq.~(S4)]{OFDM_PN_TSP_Exploitation_2023} and \eqref{eq_xi_p} into \eqref{eq_xixi}.
\end{proof}

\subsection{Location-Dependent \ac{PN} Covariance}
Based on Prop.~\ref{lemma_pn_stat_general}, the following result provides the statistics of $\bxi$.
\begin{proposition}\label{prop_xi_stat}
    The discrete-time version $\bxi \in \rset{N\Nrx}{1}$ of the \ac{PN} process $\xi(t,\tau)$ in \eqref{eq_xi_p}, which contains $N$ time-domain \ac{PN} samples across $\Nrx$ RXs, can be statistically characterized as
    \begin{align} \label{eq_pn_cov}
        \bxi\sim \mtN(\boldzero, \boldR(\pp)) \,,
        \vspace{-0.2in}
    \end{align}    
    where $\boldR(\pp) \in \rset{N\Nrx}{N\Nrx}$ is the location-dependent covariance matrix of $\bxi$. The $\thn{(m_1,m_2)}$ entry of $\boldR(\pp)$ is given by
\begin{align}\label{eq_rtau_entries}
    \left[ \boldR(\pp) \right]_{m_1,m_2} = \rxii((n_1-n_2) \Ts, \tau_{i_1}(\pp), \tau_{i_2}(\pp) ) ~,
\end{align}
where $(m_1, m_2) = (n_1 + (i_1-1) N, n_2 + (i_2-1) N)$, with $0 \leq n_1, n_2 \leq N-1$ and $1 \leq i_1, i_2 \leq \Nrx$ denoting the time-domain and RX indices, respectively, and $\Ts = T/N$ the sampling interval.
\end{proposition}
\begin{proof}
    It follows from Prop.~\ref{lemma_pn_stat_general}.
\end{proof}

\begin{remark}
    Prop.~\ref{prop_xi_stat} establishes the location-dependence of the \ac{PN} statistics in a shared-oscillator multistatic sensing setup. This result generalizes the delay-dependence of the \ac{PN} covariance in monostatic sensing \cite{OFDM_PN_TSP_Exploitation_2023} to multistatic sensing where multiple bistatic delays $\{\tau_i(\pp)\}_{i=1}^{\Nrx}$ induce dependence on $\pp$. To illustrate the location-dependence property, Fig.~\ref{fig_pn_var_loc} plots the mean accumulated \ac{PN} variance \cite{PN_Sensing_Privacy_ISAC_2026} over $N$ samples and $R$ receivers, corresponding to $1/(N \Nrx)\tracesmall{\boldR(\pp)}$, as a function of target location, showing substantial variations in \ac{PN} variance with changing location. 
\end{remark}

\begin{figure}[t]
	\centering
	\includegraphics[width=0.9\linewidth]{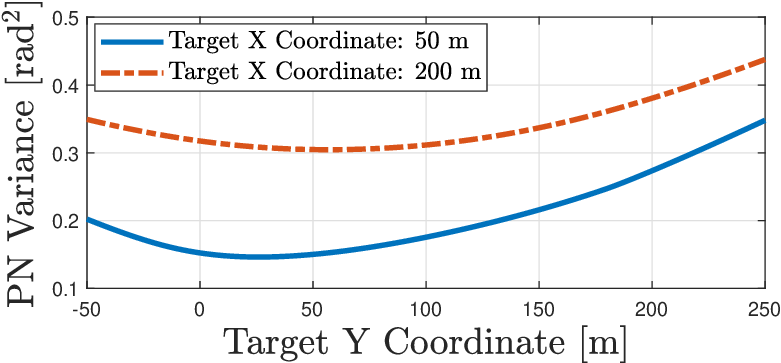}
	\vspace{0pt}
	\caption{\footnotesize Mean accumulated PN variance with respect to target location for $\fdb = 20 \, \rm{kHz}$ and the AP deployment in Fig.~\ref{fig_dep}.}
	\label{fig_pn_var_loc}
\end{figure}

\section{Localization Performance Bounds}\label{sec_bounds}
In this section, based on \eqref{eq_rx_vector_i_iso} and \eqref{eq_rx_vector_i_non_iso}, we derive the theoretical bounds on location estimation with and without PN, employing the \ac{CRB} and \ac{MCRB} tools from \cite{Fortunati2017,OFDM_PN_TSP_Exploitation_2023,PN_OFDM_ISAC_RadarConf_2023}.

\subsection{PN-Free Bound: Deterministic \ac{CRB}}
In this case, we assume an ideal \ac{LO} (i.e., without \ac{PN}, meaning $\ww_i = \boldone$ in \eqref{eq_rx_vector_i_iso} and \eqref{eq_rx_vector_i_non_iso}) to establish an ideal, PN-free benchmark for localization performance.
For \textit{isotropic} scattering, the signal model in \eqref{eq_rx_vector_i_iso} becomes
\begin{align}\label{eq_rx_vector_i_iso_crb}
    \yy_i    =    \alpha_i e^{j\vartheta}
    e^{-j2\pi f_c\tau_i(\pp)}
     \uu_i(\pp)
    +\zz_i \,,    
\end{align}
with the unknown parameters $\etab = \left[ \pp^\trp, \alphab^\trp, \vartheta \right]^\trp \in \rset{(R+3)}{1} $, where $\alphab = [\alpha_1, \ldots, \alpha_\Nrx]^\trp$. For \textit{non-isotropic} scattering, \eqref{eq_rx_vector_i_non_iso} becomes
\begin{align}\label{eq_rx_vector_i_non_iso_crb}
    \yy_i    =    \gamma_i
    \uu_i(\pp)
    +\zz_i \,,
\end{align}
with the unknown parameters $\etab = \big[ \pp^\trp, \gammab^\trp \big]^\trp \in \rset{(2R+2)}{1} $, where $\gammab = [\realp{\gamma_1}, \imp{\gamma_1}, \ldots, \realp{\gamma_\Nrx}, \imp{\gamma_\Nrx}]^\trp $. Following the same steps as in \cite[Sec.~III-A]{PN_OFDM_ISAC_RadarConf_2023}, the CRB on $\pp$ can be derived for both isotropic and non-isotropic targets.

\subsection{PN-Aware Bound: Hybrid \ac{CRB}}
In this case, we consider the generative models in \eqref{eq_rx_vector_i_iso} and \eqref{eq_rx_vector_i_non_iso} (i.e., with \ac{PN}) and derive the hybrid \ac{CRB} on location estimation \cite{hybrid_CRB_LSP_2008,Hybrid_ML_MAP_TSP}, a lower bound for a PN-aware receiver that estimates both deterministic parameters (position, gain and phase) and random PN with the statistics in \eqref{eq_pn_cov}. 
For \textit{isotropic} scattering, the signal model is given by \eqref{eq_rx_vector_i_iso}, with the unknown parameters $\etab = \left[ \pp, \alphab^\trp, \vartheta, \bxi^\trp \right]^\trp \in \rset{(N \Nrx + R+3)}{1} $. For \textit{non-isotropic} scattering, we have the signal model \eqref{eq_rx_vector_i_non_iso} and the unknown parameter vector $\etab = \left[ \pp^\trp, \gammab^\trp, \bxi^\trp \right]^\trp \in \rset{(N \Nrx + 2R+2)}{1} $. Using the prior distribution of the \ac{PN} $\bxi$ in \eqref{eq_pn_cov} and following similar steps to those in \cite[App.~A]{OFDM_PN_TSP_Exploitation_2023}, \cite{hybrid_CRB_LSP_2008}, the prior and observation related Fisher information matrices can be derived, yielding the hybrid \ac{CRB}.

\begin{figure*}[t]
    \centering

    \subfloat[Performance with respect to transmit power.%
        \label{fig_snr_sweep}]
    {%
        \includegraphics[width=0.32\textwidth]
        {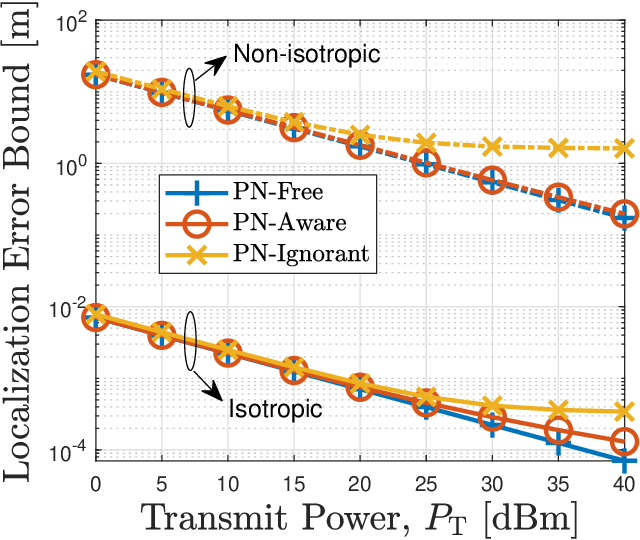}%
    }
    \hfill
    \subfloat[Performance with respect to \ac{LO} quality.%
        \label{fig_3db_sweep}]
    {%
        \includegraphics[width=0.32\textwidth]
        {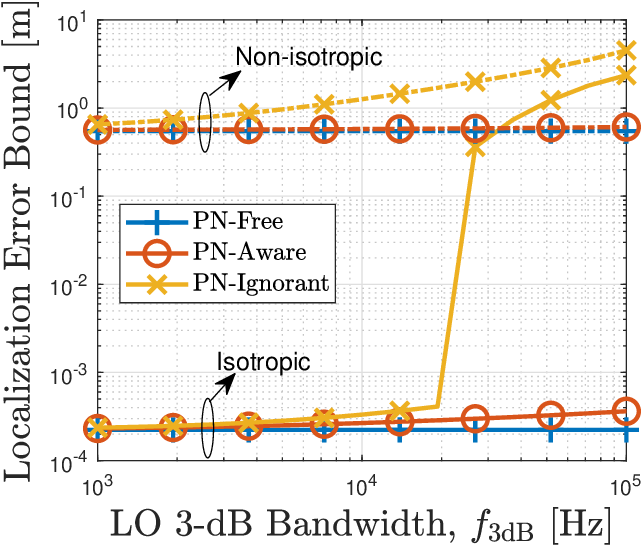}%
    }
    \hfill
    \subfloat[Impact of AP deployment scale factor $s$.
        \label{fig_loc_sweep}]
    {%
        \includegraphics[width=0.32\textwidth]
        {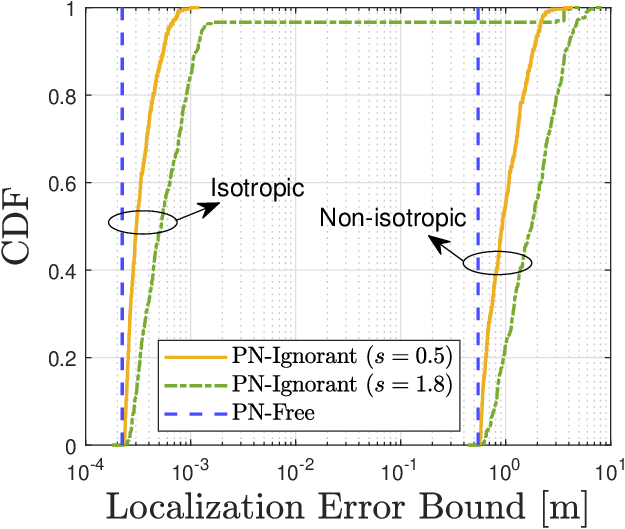}%
    }

    \caption{Bounds on localization accuracy under various parameter settings.}
    \label{fig_parameter_sweeps}
\end{figure*}

\subsection{PN-Ignorant Bound: Misspecified \ac{CRB}}
Unlike the previous case, we consider a PN-ignorant receiver that ignores the presence of \ac{PN} and estimates target location by assuming the ideal, PN-free observation model in \eqref{eq_rx_vector_i_iso_crb} or \eqref{eq_rx_vector_i_non_iso_crb}. The goal is to quantify performance loss under \ac{PN}, resulting from mismatch between the true generative models in \eqref{eq_rx_vector_i_iso} or \eqref{eq_rx_vector_i_non_iso} and the assumed models in \eqref{eq_rx_vector_i_iso_crb} or \eqref{eq_rx_vector_i_non_iso_crb}, by employing the \ac{MCRB} machinery \cite{Fortunati2017}. 

For \textit{isotropic} scattering, the true model with \ac{PN} is given by \eqref{eq_rx_vector_i_iso} for a given \ac{PN} realization $\bxi$ while the assumed model without \ac{PN} is \eqref{eq_rx_vector_i_iso_crb}. The pseudo-true parameter, which minimizes the \ac{KL} divergence between the true and assumed \ac{PDF} of the observations in \eqref{eq_rx_vector_i_iso} and \eqref{eq_rx_vector_i_iso_crb} \cite{Fortunati2017}, can be computed as \cite[Sec.~III-C]{PN_OFDM_ISAC_RadarConf_2023}
\begin{align} \label{eq_pt_min}
		(\pp_0, \alphab_0, \vartheta_0) &= \arg \min_{\pp, \alphab, \vartheta} ~ \norm{\mubbar - \mub(\pp, \alphab, \vartheta)}^2 \,,
\end{align}
where $\mubbar = [\mubbar_1^\trp, \ldots, \mubbar_\Nrx^\trp]^\trp $ with $\mubbar_i \triangleq \alphabar_i e^{j\varthetabar}
    e^{-j2\pi f_c\tau_i(\ppbar)} \ww_i \odot
     \uu_i(\ppbar)$ and $\ppbar, \{\alphabar_i\}_{i=1}^\Nrx, \varthetabar$ denoting the true values of the unknown parameters $\pp, \{\alpha_i\}_{i=1}^\Nrx, \vartheta$, and $\mub(\pp, \alphab, \vartheta) = [\mub_1^\trp(\pp, \alphab, \vartheta), \ldots, \\§ \mub_\Nrx^\trp(\pp, \alphab, \vartheta)]^\trp $ with $\mub_i(\pp, \alphab, \vartheta) \triangleq \alpha_i e^{j\vartheta}
    e^{-j2\pi f_c\tau_i(\pp)} 
     \uu_i(\pp)$. After compressing away $\alphab$ and $\vartheta$ in \eqref{eq_pt_min} via their closed-form estimates as a function of $\pp$, the concentrated cost function reduces to a 2-D search over $\pp$:
     \begin{align} \label{eq_pt_min2}
		\pp_0 &= \arg \max_{\pp} ~ \sum_{i=1}^{\Nrx} \abss{ \beta_i(\pp) }^2 + \absss{\sum_{i=1}^{\Nrx}\beta_i(\pp)^2} \,,
\end{align}
where $\beta_i(\pp) \triangleq e^{j2\pi f_c\tau_i(\pp)} \uu_i^\herm(\pp) \mubbar_i$.

For \textit{non-isotropic} scattering, the true and assumed models are given by \eqref{eq_rx_vector_i_non_iso} and \eqref{eq_rx_vector_i_non_iso_crb}, respectively. The pseudo-true parameter optimization over $\pp_0, \gammab_0$ can be formulated similarly to \eqref{eq_pt_min} and the compressed cost function over $\pp$ becomes
\begin{align} \label{eq_pt_min3}
		\pp_0 &= \arg \max_{\pp} ~ \sum_{i=1}^{\Nrx} \abss{ \zeta_i(\pp) }^2   \,,
\end{align}
where $\zeta_i(\pp) \triangleq  \uu_i^\herm(\pp) (\gammabar_i \ww_i \odot  \uu_i(\ppbar))$, with $\gammabar_i$ denoting the true value of $\gamma_i$. We note that the coherent sum in \eqref{eq_pt_min2} disappears in \eqref{eq_pt_min3} due to non-isotropic scattering preventing phase-coherent processing. In both cases, after obtaining the pseudo-true parameters via \eqref{eq_pt_min2} or \eqref{eq_pt_min3}, the derivations of the \ac{MCRB} and the \ac{LB} follow those in \cite[Sec.~III-C]{PN_OFDM_ISAC_RadarConf_2023}.

\section{Simulation Results}
For simulations, we consider the default parameters in Table~\ref{tab_parameters} and the system geometry in Fig.~\ref{fig_dep}. In \eqref{eq_rx_vector_i}, we employ QPSK modulation for $\xx$ and generate the channel gains using the bistatic radar equation \cite[Eq.~(2.9)]{Scheer2010RadarEquation} as $\alpha_i = \lambda \sqrt{\Pt \gaint \gainr \sigmarcsi}/ ((4 \pi)^{3/2} \distt \distri)$,  where $\distt = \norm{\ptx - \pp}$, $\distri = \norm{\prxi - \pp}$, $\lambda = c/\fc$ is the wavelength, $\Pt$ is the transmit power, $\gaint, \gainr$ are the transmit and receive antenna gains, and $\sigmarcsi$ denotes the bistatic \ac{RCS} of the target with respect to RX $i$. In \eqref{eq_rx_vector_i_iso} and \eqref{eq_rx_vector_i_non_iso}, a single random variable for $\vartheta$ and $\Nrx$ random variables for $\{\vartheta_i\}_{i=1}^{\Nrx}$ are drawn randomly from $\mathcal{U}[0, 2\pi)$, respectively. In addition, we set $\sigma^2 = N_0 N \deltaf N_F$ with $N_0$ and $N_F$ representing the noise \ac{PSD} and the noise figure, respectively. To evaluate multistatic sensing performance under \ac{PN}, we investigate the three bounds in Sec.~\ref{sec_bounds} along various sweep dimensions (namely, transmit power, \ac{LO} quality and deployment). The PN-ignorant bound (i.e., the \ac{LB}) is computed by averaging over $300$ PN realizations\footnote{We numerically solve \eqref{eq_pt_min2} over a disc of radius $5$ m centered at the true position using a coarse-to-fine grid search followed by constrained \ac{SQP} refinement. For \eqref{eq_pt_min3}, we apply grid-initialized constrained \ac{SQP} over the AP bounding rectangle expanded by a factor of $1.25$ about its center.}.

\begin{table}[t]\footnotesize
\caption{Simulation Parameters}
\centering
\begin{tabular}{ll}
\hline
Carrier frequency & $\fc = 3.5$ GHz\\
Subcarriers / spacing & $N=64$, $\deltaf=60$ kHz\\
Useful / CP duration & $T=16.66~\mu$s, $\Tcp=2.14~\mu$s\\
Oscillator quality & $\fdb =20$ kHz\\
Transmit power / antenna gains & $\Pt = 30$ dBm, $\gaint = \gainr = 4$ dBi \\
Target \ac{RCS} & $\sigmarcsi = 10$ dBsm, $\forall i$\\
Noise figure / \ac{PSD} & $N_F = 8$ dB, $N_0 = -174$ dBm/Hz \\
\hline
\end{tabular}
\label{tab_parameters}
\end{table}

\begin{figure}[t]
	\centering
	\includegraphics[width=0.4\linewidth]{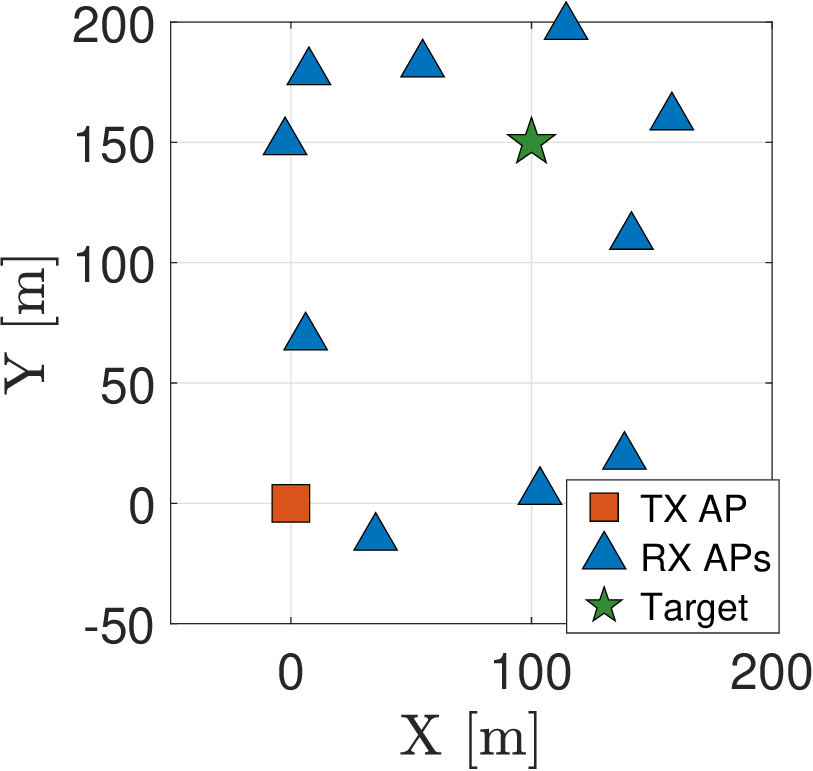}
	\vspace{0pt}
	\caption{AP deployment and target location.} 
	\label{fig_dep}
\end{figure}

\subsection{Impact of Transmit Power}
Fig.~\ref{fig_snr_sweep} reports the localization performance as a function of \ac{TX} power $\Pt$. For both target models, the PN-ignorant bound reaches a performance floor at high $\Pt$, where \ac{PN} dominates \ac{AWGN} and model mismatch becomes the main limiting factor. For non-isotropic targets, the PN-aware receiver can recover a larger portion of the \ac{PN}-induced performance loss than for isotropic targets. This is because the carrier phase information available for isotropic targets is inherently more sensitive to residual phase perturbations caused by \ac{PN}. 

\subsection{Impact of Oscillator Quality}
Fig.~\ref{fig_3db_sweep} shows the bounds as the \ac{LO} bandwidth $\fdb$ increases, corresponding to worsening \ac{LO} quality, and Table~\ref{tab_rare_tail} complements the figure by providing the percentage of PN realizations resulting in PN-ignorant bound values above $1$ m for isotropic scattering. The PN-aware bound remains close to the PN-free baseline for non-isotropic scattering, while isotropic scattering exhibits a modest increase. The effect is more pronounced when \ac{PN} is ignored. The non-isotropic bound deteriorates gradually, whereas the isotropic bound rises abruptly at larger $\fdb$ values and approaches the non-isotropic bound. As seen from Table~\ref{tab_rare_tail}, this jump results from PN-induced shifts of the pseudo-true position $\pp_0$ in \eqref{eq_pt_min2} away from the true position $\pp$, creating outliers increasing in percentage with $\fdb$. This implies that severe PN may erase the carrier phase advantage of isotropic targets.

\begin{table}[t]
\centering
\caption{Percentage of PN realizations yielding above-meter PN-ignorant bound values for isotropic scattering in
Fig.~\ref{fig_3db_sweep}.}
\label{tab_rare_tail}

\footnotesize
\setlength{\tabcolsep}{2.5pt}
\renewcommand{\arraystretch}{1.2}

\begin{tabular}{
    >{\centering\arraybackslash}m{0.25\columnwidth}|
    *{4}{>{\centering\arraybackslash}m{0.14\columnwidth}}
}
\hline
$\fdb$
    & $19.3$ kHz
    & $37.2$ kHz
    & $71.2$ kHz
    & $100$ kHz
\\
\hline
\shortstack[c]{Percentage of \\ PN realizations}
    & 0\%
    & 3.7\%
    & 21\%
    & 35.7\%
\\
\hline
\end{tabular}
\end{table}

\subsection{Impact of Deployment}
To investigate the impact of the location-dependent PN statistics in \eqref{eq_pn_cov}, we radially scale all AP positions about the fixed target by a deployment scale factor $s$, with $s=1$ denoting the original deployment. We keep the channel gains at their $s=1$ values, thereby keeping the PN-free bound fixed\footnote{Radial scaling about the
fixed target preserves the gradients of the bistatic delays with respect to the target location.} while varying the PN statistics through the bistatic delays. Fig.~\ref{fig_loc_sweep} shows the empirical \ac{CDF} of the PN-ignorant bound across PN realizations for $s=0.5$ (contraction) and $s=1.8$ (expansion), together with the PN-free baseline. For both scattering models, expansion shifts the CDF toward larger bound values, indicating more severe PN-induced degradation and confirming location-dependent nature of PN. Specifically, the longer bistatic delays for $s=1.8$ increase the differential PN variance in \eqref{eq_dpn_stat}, consistent with the location-dependent PN effects illustrated in Fig.~\ref{fig_pn_var_loc}. For isotropic scattering, expansion also produces a distinct upper tail with meter-level bound values. This suggests that, for some PN realizations, the carrier phase localization advantage is effectively lost, yielding bounds comparable to those under non-isotropic scattering.

\section{Conclusion}

We have characterized the location-dependent PN statistics in shared-\ac{LO} multistatic OFDM sensing and evaluated localization accuracy using PN-free, PN-aware and PN-ignorant bounds. The results reveal error floors under PN-ignorant processing and a loss of the accuracy gain provided by carrier phase information under severe PN for isotropic targets, while the PN-aware bounds indicate significant potential for mitigation. The impact of location-dependent PN statistics has been shown via simulations under different deployments. Future research will consider \acp{PLL} and joint location-velocity estimation with multiple symbols.



\balance
\bibliographystyle{Sub/IEEEbib}
\bibliography{IEEEabrv,Sub/isac}
\end{document}